\documentclass[letterpaper, 10 pt, conference]{ieeeconf}  

\IEEEoverridecommandlockouts                              

\usepackage{amsmath} 
\usepackage{amssymb}  

\usepackage{cite}

\usepackage{algorithm}
\usepackage{algpseudocode}
\usepackage{graphicx}

\usepackage[caption=false, font=footnotesize]{subfig}

\usepackage[usenames, dvipsnames]{color}

\usepackage{flushend}

\newtheorem{theorem}{Theorem}
\newtheorem{corollary}[theorem]{Corollary}
\newtheorem{lemma}[theorem]{Lemma}
\newtheorem{definition}{Definition}
\newtheorem{assumption}{Assumption}
\newtheorem{remark}{Remark}

\newcommand{\dist}{{\textnormal{dist}}}

\newcommand{\est}[1]{\widehat{#1}}
\newcommand{\trunc}[1]{\underline{#1}}
\newcommand{\qedsym}{\hfill\blacksquare}

\newboolean{showcomments}
\setboolean{showcomments}{false}

\newcommand{\todo}[1]{  \ifthenelse{\boolean{showcomments}}
{\textcolor{ForestGreen}{TO DO:  #1}}{}}
\newcommand{\lina}[1]{\ifthenelse{\boolean{showcomments}}
{\textcolor{Orange}{(Lina says: #1)}}{}}
\newcommand{\johan}[1]{\ifthenelse{\boolean{showcomments}}
{\textcolor{Blue}{(Johan says: #1)}}{}}
\newcommand{\runyu}[1]{\ifthenelse{\boolean{showcomments}}
{\textcolor{Red}{(Runyu says: #1)}}{}}
\newcommand{\emma}[1]{\ifthenelse{\boolean{showcomments}}
{\textcolor{VioletRed}{(Emma says: #1)}}{}}
\newcommand{\fix}[1]{\ifthenelse{\boolean{showcomments}}
{\textcolor{Blue}{#1}}{{#1}}}
\usepackage{setspace}
\newboolean{shortversion}
\setboolean{shortversion}{false}

\title{\LARGE \bf
Scalable Reinforcement Learning for Linear-Quadratic \\ Control of Networks
}

\author{Johan Olsson, Runyu (Cathy) Zhang, Emma Tegling, and Na Li
\thanks{J. Olsson and E. Tegling are with the Department of Automatic Control and the ELLIIT Strategic Research Area at Lund University, Sweden. {\tt\small jo1272ol-s@student.lu.se, tegling@control.lth.se.} R. Zhang and N. Li are with the John A. Paulson School of Engineering and Applied Sciences, Harvard University, MA. {\tt\small runyuzhang@fas.harvard.edu, nali@seas.harvard.edu.} 
J. Olsson carried out this research while being a visiting graduate student at the John A. Paulson School of Engineering and Applied Sciences, Harvard University. }\thanks{This work was partially funded by NSF AI institute: 2112085, NSF CNS: 2003111 and the Wallenberg AI, Autonomous Systems and Software Program (WASP) funded by the Knut and Alice Wallenberg Foundation.}
}

\begin{document}

\maketitle
\thispagestyle{empty}
\pagestyle{empty}

\begin{abstract}
Distributed optimal control is known to be challenging and can become intractable even for linear-quadratic regulator problems. In this work, we study a special class of such problems where distributed state feedback controllers can give near-optimal performance. More specifically, we consider networked linear-quadratic controllers with decoupled costs and spatially exponentially decaying dynamics. We aim to exploit the structure in the problem to design a scalable reinforcement learning algorithm for learning a distributed controller. Recent work has shown that the optimal controller can be well approximated only using information from a $\kappa$-neighborhood of each agent. Motivated by these results, we show that similar results hold for the agents' individual value and Q-functions. We continue by designing an algorithm, based on the actor-critic framework, to learn distributed controllers only using local information. Specifically, the Q-function is estimated by modifying the Least Squares Temporal Difference for Q-functions method to only use local information. The algorithm then updates the policy using gradient descent. Finally, we evaluate the algorithm through simulations that indeed suggest near-optimal performance.

\end{abstract}
\ifthenelse{\boolean{shortversion}}
{
\input{Sections/shortversion/introduction-short}
\input{Sections/shortversion/prelim-short}
\input{Sections/shortversion/structure-V-Q-functions-short}
\input{Sections/shortversion/algorithm-design-short}
\input{Sections/shortversion/numerical-simulation-short}
\input{Sections/shortversion/conclusion-short}
}
{
\section{INTRODUCTION}
Multi-agent networked systems such as power grids, wireless communication networks and smart buildings have been extensively studied in recent years. Due to scalability and resilience concerns, distributed control of these systems is highly desirable, but known to be challenging in practice, even for the linear quadratic regulator (LQR). Interactions between agents are, however, often local in nature and previous work has shown that this makes distributed control feasible and even near-optimal~\cite{Bamieh2002, 4623272, shin2023near, motee2017sparsity, zhang2022optimal}. Specifically, for networked LQR, the work in~\cite{zhang2022optimal} shows that when the dynamics have a spatially exponentially decaying (SED) structure (see Definition~\ref{def:SED}), so does the optimal control policy. Furthermore, it is shown that a truncated controller, i.e., a controller that only uses local information, only gives a small, quantifiable, suboptimality gap. In many cases of large-scale multi-agent networked systems, system parameters are unknown (or only partially known). This raises a natural question of how to find such a distributed controller when the system dynamics are unknown.


In the case of centralized (single-agent) LQR, reinforcement learning algorithms have been studied extensively. In particular, much attention has been given to exploiting the quadratic structure of the value and Q-functions. See, e.g.,~\cite{bradtke1992reinforcement,9304202, krauth2019finite, fazel2018global, abbasi2019model} for a small subset of the work. Further,~\cite{jing2021learning} study centralized learning of decentralized LQR. However, for large systems or when data privacy is a concern, distributed learning schemes may be required. In general, these are harder to design and many questions still remain open regarding their performance.

On the other hand, scalable, multi-agent reinforcement learning algorithms are feasible under certain local-interaction assumptions~\cite{qu2020scalable, zhang2023global}.
In the LQR setting, distributed learning has been studied in~\cite{li2021distributed, gorges2019distributed, alemzadeh2019distributed}. In~\cite{li2021distributed}, each agent only observes a partition of the global state, utilizing consensus and derivative-free optimization to find a distributed controller. However, the algorithm relies on Monte-Carlo estimation, which may suffer from high variance. In both~\cite{gorges2019distributed} and~\cite{alemzadeh2019distributed} distributed Q-learning for LQR is studied. Nevertheless, both works consider a more restricted setting, with stronger decoupling assumptions on the dynamics.

\textbf{Contributions.} In this paper, we investigate distributed reinforcement learning for distributed LQR. Specifically, a model-free, reinforcement learning algorithm is proposed for the infinite-horizon discrete-time network LQR problem, where $N$ agents in a network together aim to minimize the long-term average cost. Our focus is on spatially truncated controllers and networks where the system matrices are assumed to be SED (see Definition~\ref{def:SED}) and the costs decoupled. In Section~\ref{sec:V_Q_SED_structure}, we show that the agents' individual value and Q-functions also exhibit spatial exponential decay. Not only does this allow us to upper bound the error of truncating these functions, but it also suggests that distributed learning of truncated controllers is feasible. Motivated by these results, we then design a distributed learning scheme based on the actor-critic framework to find truncated controllers by exploiting the spatially decaying structure of the Q-functions in Section~\ref{sec:alg_des}. Going forward, we believe that, apart from the specific algorithm considered in this paper, the spatially decaying structure of the value and Q-functions opens up possibilities for more flexible and diverse approaches to distributed algorithm design. 


\vspace{4pt}
\noindent\textbf{Notation.} We let~$||\cdot||$ denote both the~$l_2$-norm of a vector and the induced~$l_2$-norm of a matrix. For a symmetric matrix~$M \in \mathbb{R}^{n \times n}$,~$\text{svec}(M) \in \mathbb{R}^{n(n+1)/2}$ denotes the vectorized version of the upper triangular part of $M$ so that~$\text{svec}(M)^{\top}\text{svec}(M) = ||M||^2_F$. We let~$\text{smat}(\cdot)$ be the inverse of~$\text{svec}(\cdot)$ so that~$\text{smat}(\text{svec}(M)) = M$. Furthermore,~$\est{\cdot}$ denotes estimation from samples and~$\trunc{\cdot}$ truncation.
\section{Preliminaries \& Problem Setup}

\subsection{Network LQR for SED Systems}
Consider an infinite-horizon, discrete-time, network LQR problem with $N$ agents,~$[N] := \{1,\dots, N \}$, embedded on an undirected graph. The graph is equipped with a distance function $\dist(\cdot, \cdot): [N] \times [N] \rightarrow \mathbb{R}_{\ge 0}$, for which~$\dist(i,j) = \dist(j,i)$ and the triangle inequality, ~$\dist(i,j) \leq \dist(i,k) + \dist(k,j)$ holds for all~$i,j,k \in [N]$. Since we are considering problems where the agents are embedded on an undirected graph, we let $\dist(\cdot, \cdot)$ refer to the graph distance, i.e., the shortest distance between any two nodes in the graph. We will, however, keep in mind that our results hold for all distance functions. The graph distance allows us to introduce the concept of the~$\kappa$-neighborhood of agent~$i$,
\begin{equation*}
    \mathcal{N}_i^\kappa := \{ j \in [N], \dist(i,j) < \kappa\}.
\end{equation*}

With the graph in place, we now turn to the model dynamics and cost function. The global state $x$ and control action $u$ are given by 
\begin{align*}
    x(t) &= [x_1(t)^{\top}, x_2(t)^{\top}, \dots, x_N(t)^{\top}]^{\top} \in \mathbb{R}^n,\\
    u(t) &= [u_1(t)^{\top}, u_2(t)^{\top}, \dots, u_N(t)^{\top}]^{\top} \in \mathbb{R}^m,
\end{align*}
where~$x_i(t) \in \mathbb{R}^{n_{i}}$ and~$u_i(t) \in \mathbb{R}^{m_{i}}$ are each agent's local state and control signal, respectively, and~$n=\sum_i n_i$ and~$m=\sum_i m_i$. We also introduce the notation $x_{\mathcal{N}^\kappa_i}$, meaning the concatenation of all~$x_j$ with~$j \in \mathcal{N}_i^\kappa$ and similar for~$u_{\mathcal{N}^\kappa_i}$.

We assume linear dynamics 
\begin{equation*}
    x(t+1) = Ax(t) + Bu(t) + w(t), \  w(t) \sim \mathcal{N}(0,\sigma_w^2I),
\end{equation*}
where $w(t)$ is i.i.d. Gaussian noise. The global state and action spaces can be partitioned into local states $x_i$ and actions $u_i$, meaning that for agent~$i$, the state at time~$t+1$ is given by
\begin{equation*}
    x_i(t+1) = \sum_{j=1}^{N} [A]_{ij}x_j(t) + [B]_{ij}u_j(t)+\omega_i(t)
\end{equation*}
\fix{with $\omega_i(t) \sim \mathcal{N}(0,\sigma_w^2I)$}. Here,~$[M]_{ij}$ for a matrix $M$ denotes the submatrix of~$M$ where the row indices correspond to the indices of agent~$i$ and the column indices correspond to the indices of agent~$j$.\footnote{By indices of agent~$i$, if the total index length is~$n=\sum_i n_i$, we mean the indices in the range~$[\sum_{j=1}^{i-1}n_j + 1, \sum_{j=1}^{i} n_j]$.} Furthermore, we let~$[M]_{i:}$ and~$[M]_{:i}$ denote the set of rows and columns corresponding to agent~$i$ respectively. 

For each agent, there is also a quadratic local cost which only depends on the state and action of the agent itself
\begin{equation*}
    c_i(t) = x_i(t)^{\top}[S]_{ii}x_i(t) + u_i(t)^{\top}[R]_{ii}u_i(t),
\end{equation*}
with~$[S]_{ii} \succeq 0 \in \mathbb{R}^{n_i \times n_i}$ and~$[R]_{ii} \succ 0\in \mathbb{R}^{m_i \times m_i}$. The global cost is defined as the summation of the individual costs
\begin{equation*}
    c(t) := \sum_{i=1}^{N} c_i(t) = x(t)^{\top}Sx(t) + u(t)^{\top}Ru(t).
\end{equation*}
With~$S \succeq 0 \in R^{n \times n}$ and~$R \succ 0 \in R^{m \times m}$ both block-diagonal.
Restricting ourselves to static linear feedback policies of the form~$u(t)~=~Kx(t)$, the problem can be formulated as a classical LQR problem
\begin{align}
\label{eq:LQR_P}
&\min_{K} \underbrace{\lim_{T \to \infty} \mathbb{E}_{} \left [ \frac{1}{T}\sum_{t=0}^{T-1} c(t) \right ]}_{J(K)}, \\ 
&\mathrm{s.t.}~x(t+1)\! =\! Ax(t) \!+\! Bu(t) \!+\! w(t),~~w(t) \sim \mathcal{N}(0,\sigma_w^2I) \notag, \\
  & \ \ \ \ \ u(t) = Kx(t). \notag
 \end{align}

In this work, we focus on stabilizing policies and we define the concept of~$(\tau, \rho)$-stability.
\begin{definition}[$(\tau, \rho)$-stability] \label{def:stability}
    For~$\tau \geq 1, \ \rho > 0$, a matrix~$X$ is said to be~$(\tau, \rho)$-stable if~$||X^k|| \leq \tau \cdot e^{-\rho k}$, for all $k \in \mathbb{Z}_{\geq 0}$.
\end{definition}
We say that~$K$ is stabilizing if there exist~$\tau,~\rho$ such that the closed-loop system,~$A+BK$, is~$(\tau, \rho)$-stable.

Moreover, due to sensing and communication constraints, the agents in large-scale networks often have to take control actions based on local observations only. We consider the setting where agents can only observe state information within their $\kappa$-neighborhood. This motivates us to study a special class of $\kappa$-truncated policies, defined by
\begin{equation*}
    \mathcal{K}^\kappa := \{ K \in \mathbb{R}^{m \times n}: [K]_{ij} = 0_{m_i \times n_j} \text{ if } j \notin \mathcal{N}_i^\kappa  \}.
\end{equation*}
\noindent For general network systems, such local policies could lead to poor performance compared to the optimal global controller. However, when the system has a certain spatially decaying structure, previous work has shown that $\kappa$-truncated control can achieve near-optimal performance~\cite{zhang2022optimal}. The goal of this paper is to design learning methods to find these near-optimal local controllers when the spatially decaying structure holds.

\vspace{4pt}
\noindent\textbf{Spatially Exponentially Decaying (SED) Structure} The problem we consider here is a special type of LQR problem in which the individual agents' costs have been decoupled and where the dynamics are unknown but satisfy a spatially decaying structure as first considered in~\cite{zhang2022optimal}. 

\begin{definition}[Spatial exponential decay (SED)]  \label{def:SED}
    Given a matrix $X \in \mathbb{R}^{\sum_{i=1}^N n_i \times \sum_{i=1}^N m_i}$ partitioned into $N \times N$ blocks, $[X]_{ij} \in \mathbb{R}^{n_i \times m_j}$, and distance function~$\dist(\cdot, \cdot): [N] \times [N] \rightarrow \mathbb{R}_{\ge 0}$, the block matrix $X$ is $(c,\gamma)-SED$ if
    \begin{equation*}
    \left \| [X]_{ij} \right \| \leq c \cdot e^{-\gamma \dist(i,j)}, \ \forall \ i,j \in [N].
\end{equation*}
\end{definition}
\vspace{2mm}

\noindent The purpose of Definition~\ref{def:SED} is to quantify the rate of decay in the interaction between interconnected agents and we make the following assumption on the dynamics.


\begin{assumption}[SED dynamics] \label{ass:SED} There exist~$\gamma_{sys} > 0$ and constants~$c_A,c_B > 0$ such that $A, B$ are~$(c_A, \gamma_{sys})$-SED and $(c_B, \gamma_{sys})$-\textnormal{SED}, respectively. Without loss of generality, we assume~$c_A,c_B \geq 1$.
\end{assumption}

When later discussing individual value and Q-functions, we will be using a similar concept to SED that quantifies a specific agent's importance for a given matrix. Thus, we define the concept of spatial decay away from an agent:

\begin{definition} \label{def:SED_away_i} Given a matrix $X \in \mathbb{R}^{\sum_i n_i \times \sum_i m_i}$ partitioned into $N \times N$ blocks, $[X]_{ij} \in \mathbb{R}^{n_i \times m_j}$, and a distance function~$\dist(\cdot, \cdot): [N] \times [N] \rightarrow \mathbb{R}_{\ge 0}$, the block matrix $X$ is $(c,\gamma)$\textit{-SED away from~i}, if~$i \in [N]$ and
\begin{equation*}
    ||[X]_{lj}|| \leq c \cdot e^{-\gamma \max(\dist(i,l), \dist(i,j))}, \ \forall \ l,j \in [N].
\end{equation*}
\end{definition}
\vspace{2mm}

\begin{remark}
Definition~\ref{def:SED_away_i} is stronger than Definition~\ref{def:SED} in the sense that if $X$ is $(c,\gamma)$-SED away from~$i$, then it is also $(c,\gamma /2)$-SED. This can be seen by combining the triangle inequality with the fact that the maximum of two numbers is greater than or equal to their average.
\end{remark}

The SED property becomes relevant and useful when $c$ is small and $\gamma$ large relative to the matrix (network) size, in particular as this grows large. While it is true that all finite-dimensional matrices fulfill Definition~\ref{def:SED} for some $\gamma$ and $c$, the question of interest is how the SED property carries over from the dynamics to, e.g., the optimal controller.

\vspace{6pt}
\noindent\textbf{Problem Statement.} We consider reinforcement learning for network LQR with system matrices that are SED. At each timestep, agent~$i$ observes its own state~$x_i(t)$ and cost~$c_i(t)$. Agents can also communicate their state and control action information with their~$\kappa$-neighborhood at each timestep. The goal is to exploit the spatially decaying structure of the system, in order to design a distributed learning algorithm that finds a $\kappa$-truncated control policy, i.e., $u(t) = Kx(t)$ with $K \in \mathcal{K}^\kappa$, that minimizes the global cost function~$J(K)$ in~\eqref{eq:LQR_P}.

\subsection{Preliminaries: Value Functions and Q-functions}
The distributed learning algorithm in Section~\ref{sec:alg_des} will exploit structure in the Q-function. In order to explain the algorithm, we briefly introduce the value function and Q-function. 
For a given controller~$u(t) = Kx(t)$, we define the value function
\begin{equation*}
    V^K(x) := \mathbb{E} \left[ \sum_{t=0}^\infty (c(x(t),Kx(t)) - \lambda^K) \mid x(0) = x \right],
\end{equation*}
where~$\lambda^K$ is the expected average stage cost of policy~$K$ under stationarity, and~$\lambda^K := \lim_{T \rightarrow \infty} \mathbb{E} \left[ \frac{1}{T} \sum_{t=0}^T c(x(t), Kx(t)) \right ]$. Similarly, the cost of taking an arbitrary action~$u$ from state~$x$ and thereafter follow~$K$ is given by the Q-function
\begin{equation*}
    Q^K(x,u) \!:=\! \mathbb{E} \! \Bigg[   \sum_{t=0}^\infty (c(x(t), u(t)) - \lambda^K ) \! \mid \! x(0) \!= \!x, u(0) \!=\! u  \Bigg].
\end{equation*}
It is well known that these functions are quadratic for the LQR problem. Specifically, it holds that \begin{align*}
    V^K(x) &= x^{\top}Px, \\
    Q^K(x,u) &= (x^{\top}, u^{\top})H \begin{pmatrix}
x\\ u
\end{pmatrix}\! =\! \text{svec}(H)^\top  \text{svec} \left(\!   \begin{pmatrix}
x\\ u
\end{pmatrix} \!\!\begin{pmatrix}
x\\ u
\end{pmatrix}^{\!\!\!\top} \!\right),
\end{align*}
where~$P$ is the solution to the Lyapunov equation,
\begin{equation*}
    P = S + K^{\top} RK + (A+BK)^{\top} P (A+BK),
\end{equation*}
and
\begin{equation*}
    H = \begin{pmatrix}
H_{11} &  H_{12}\\ 
H^{\top}_{12} & H_{22}
\end{pmatrix} = \begin{pmatrix}
S + A^{\top}PA &  A^{\top} PB\\ 
B^{\top}PA & R + B^{\top} PB
\end{pmatrix}.
\end{equation*}
Apart from the value function and Q-function defined for the global cost $c$, we also define the individual value function $V_i$ and Q-function $Q_i$ for agent $i$'s local cost $c_i$ as
\begin{align*}
    V^K_i(x) &:= \mathbb{E} \left[ \sum_{t=0}^\infty (c_i(x(t),Kx(t)) - \lambda_i^K) \mid x(0) = x \right] \!, \\
     Q^K_i\!(x,\!u) &:= \mathbb{E} \Bigg[\! \sum_{t=0}^\infty \! (c_i(x(t), u(t)) \!-\! \lambda_i^K ) \!\mid\! x(0) \!=\! x, u(0) \!=\! u  \! \Bigg] \! ,
\end{align*}
where $\lambda_i^K := \lim_{T \rightarrow \infty} \mathbb{E} \left[ \frac{1}{T} \sum_{t=0}^T c_i(x(t), Kx(t)) \right ]$. We let~$S_i \in \mathbb{R}^{n \times n}$ and~$R_i \in \mathbb{R}^{m \times m}$ be zero-padded versions of~$[S]_{ii}$ and~$[R]_{ii}$, such that,~$x^{\top}S_ix = x_i^{\top}[S]_{ii}x_i$ and~$u^{\top}R_iu = u_i^{\top}[R]_{ii}u_i$. Then, $V_i$ and $Q_i$ can be represented using the matrices $P_i, H_i$, i.e.,
\begin{align*}
    V_i^K(x) = x^\top P_i x,
    \quad Q_i^K(x,u) = \begin{pmatrix}
        x^{\top} & u^{\top}
    \end{pmatrix}H_i \begin{pmatrix}
        x\\ u
    \end{pmatrix} ,
\end{align*}
where
\begin{align}
    P_i \! &= \! S_i \!+\! [K]_{i:}^{\top} [R]_{ii}[K]_{i:} + (A\!+\!BK)^{\top} P_i (A\!+\!BK),
    \label{eq:P_i}\\
     H_i \! &= \! \begin{pmatrix}
H_{i11} & H_{i12} \\ 
H^{\top}_{i12} & H_{i22}
\end{pmatrix} \! = \! \begin{pmatrix} \!
S_i \!+ \!A^{\top}P_iA \! & \! A^{\top}P_iB \\ 
B^{\top}P_iA \! & \! R_i \!+ \!B^{\top}P_iB \!
\end{pmatrix}. \label{eq:H_i_sub_def}
\end{align}
\vspace{1mm}
\section{Spatial Decay for the Individual Value and Q-functions}
\label{sec:V_Q_SED_structure}
Value functions and Q-functions play an important role in our learning algorithm's design. In general, however, the individual Q-functions depend on the global state and control input. Thus, in order to facilitate distributed learning, it is important to study how well these individual Q-functions can be approximated only using state information from within the $\kappa$-neighborhood. For this purpose, we now investigate the SED structure of individual value functions~$V_i$ and Q-functions $Q_i$. The idea is that, if most of the useful information is contained within the agents' neighborhood, then these functions can be well approximated even if information from other agents is discarded. 

We first notice that~$P_i$ solving~\eqref{eq:P_i} plays an important role for both the individual value and Q-functions. Hence, we commence our investigation by examining how the solution to a Lyapunov equation preserves spatial decay.

\begin{lemma} \label{lemma:P_decay}
    Let~$L \in \mathbb{R}^{n \times n}$ be~$(\tau, \rho)$-stable and~$(c_L, \gamma)$-SED, 
    with~$c_L \geq 1$, and~$M \in \mathbb{R}^{n \times n}$ $(c_{M},\gamma)$\textit{-SED away from i},
    then the solution~$P$ to the Lyapunov equation~$P =L^{\top}PL + M$, is $(c_{P}, \gamma_{P})$\textit{-SED away from i} with
    \begin{equation*}
        c_{P} = \frac{||M|| \tau^2}{1- e^{-2\rho}} + 2c_{M}, \quad \gamma _{P} = \frac{\rho \gamma}{\rho + \ln(Nc_L)}.
    \end{equation*}
\end{lemma}
\vspace{2mm}

Lemma~\ref{lemma:P_decay} states that the Lyapunov equation preserves the spatial decay away from~$i$ when the matrix~$L$ is~$(\tau, \rho)$-stable. The proof is essentially borrowed from~\cite{zhang2022optimal} where a similar Lyapunov equation was studied, can be found in Appendix~\ref{apdx:proof_P_decay_lemma}.

\begin{remark}
In order for Lemma~\ref{lemma:P_decay} to provide a useful bound, it is important that the parameters in Lemma~\ref{lemma:P_decay} do not scale with~$N$. Due to the~$\ln(Nc_L)$ term, the exponent always scales with~$N$ which worsens the bound as~$N$ grows. However, similar to as argued in~\cite{zhang2022optimal}, as long as~$\max (\dist(i,l), \dist(i,j)) \geq N^\epsilon$ given any constant~$\epsilon > 0$, then~$[P]_{lj} \rightarrow 0$ as~$N \rightarrow \infty$ \fix{since~$\lim_{N\rightarrow \infty} \frac{N^\epsilon}{\ln{N}} = \infty $}.
\end{remark}

We now turn to our main result which concerns the specific Lyapunov equation~\eqref{eq:P_i} when $K$ is a $\kappa$-truncated. The result follows almost immediately when applying Lemma~\ref{lemma:P_decay} to~\eqref{eq:P_i} and the proof is deferred to Appendix~\ref{apdx:proof_V_Q_decay}.
\begin{theorem} \label{thm:value_ASED}
Let~$K \in \mathcal{K}^\kappa$ be \fix{$(c_K, \gamma_{sys})$-SED and} such that the closed-loop system,~$A+BK$, is~$(\tau, \rho)$-stable. Then, for the solution~$P_i$ to the Lyapunov equation
\begin{equation*}
    P_i = S_i + [K]_{i:}^{\top} [R]_{ii}[K]_{i:} + (A+BK)^{\top} P_i (A+BK),
\end{equation*}
it holds that $P_i$ is $(c_{P_i}, \gamma_{P_i})$\textit{-SED away from i}, with
\begin{equation*}
        c_{P_i} = \frac{||S_i + [K]_{i:}^{\top} [R]_{ii}[K]_{i:}|| \tau^2}{1- e^{-2\rho}} + 2(||[S]_{ii}|| +  ||[R]_{ii}||c_K^2),
    \end{equation*}
    and
    \begin{equation*}
       \gamma _{P_i} = \frac{\rho \gamma_{sys} }{\rho + \ln(Nc_A + N^2c_Bc_K)}.
    \end{equation*}
\end{theorem}
\vspace{2mm}

Similarly, we are interested in the decay rate of~$H_i$ from~\eqref{eq:H_i_sub_def} that parameterize agent~$i$'s individual Q-function. The proof is found in Appendix~\ref{apdx:proof_V_Q_decay}.

\begin{corollary} \label{corollary:Q_ASED}
    For a linear policy fulfilling the assumptions in Theorem~\ref{thm:value_ASED},  the submatrices~$H_{i11}, H_{i12}$ and~$H_{i22}$ of the matrix~$H_i$ defined in~\eqref{eq:H_i_sub_def} are all $(c_{H_i},\gamma_{P_i})$\textit{-SED away from~i}, with
    \begin{equation*}
        c_{H_i} \!=\! \max ( ||S_i|| + N^2c_A^2c_{P_i}, 
    N^2c_Ac_Bc_{P_i}, ||R_i|| + N^2c_B^2c_{P_i} ).
    \end{equation*}
\end{corollary}
\vspace{2mm}
Corollary \ref{corollary:Q_ASED} hints that $Q_i$ can be well approximated using the local state and action of agent $i$'s $\kappa$-neighborhood. To capture the approximation accuracy, we first define the following $\kappa$-truncation operation

\begin{definition}Let $X$ be an arbitrary matrix. The truncation of $X$, $\trunc{X}^\kappa$ is defined by $[\trunc{X}^\kappa]_{lj}=[X]_{lj}$ if $\max(\dist(i,l), \dist(i,j)) < \kappa$ and 0 otherwise.
    \label{def:trunc_X_kappa}
\end{definition}

Using Corollary~\ref{corollary:Q_ASED}, we can now bound the error caused by truncating $H_{i11}$,~$H_{i12}$ and~$H_{i22}$ with the following corollary (see Appendix~\ref{apdx:proof_H_trunc_err} for the proof).

\begin{corollary} \label{corr:H_trunc_err}
For a linear policy fulfilling the assumptions in Theorem~\ref{thm:value_ASED}, the error caused by truncating the submatrices of~$H_i$ is bounded by
    \begin{equation*}
        ||H_{i11} \! - \! \trunc{H}^\kappa_{i11} ||, ||H_{i12} \! - \! \trunc{H}^\kappa_{i12} ||, ||H_{i22} \! - \! \trunc{H}^\kappa_{i22} || \! \leq \! \sqrt{N} \! c_{H_i} \! e^{-\gamma_{P_i} \kappa}.
    \end{equation*}
\end{corollary}
\vspace{2mm}


\noindent Corollary~\ref{corr:H_trunc_err} implies that in order to truncate the submatrices of~$H_i$ with an~$\epsilon > 0$ error, taking~$\kappa \geq \ln \left( \frac{\sqrt{N}c_{H_i}}{\epsilon} \right) / \gamma_{P_i}$ is sufficient. Thus,  as long as~$c_{H_i}$ does not grow exponentially with~$N$, which is generally not the case of interest, Corollary~\ref{corr:H_trunc_err} can be used to give a reasonably small bound on the error due to truncation.


In conclusion, we have seen that agent~$i$'s individual value and Q-functions both exhibit exponential decay and that the error due to truncating the latter can be upper bounded. We now turn to designing an algorithm that leverages this by learning truncated individual Q-functions.
\section{Algorithm Design}
\label{sec:alg_des}
In this section, we propose a model-free reinforcement learning algorithm based on the actor-critic framework for learning $\kappa$-truncated controllers. The critic component of the algorithm is based on the least-squares temporal difference learning (LSTDQ) method introduced in~\cite{lagoudakis2003least} and later studied in the LQR setting in~\cite{krauth2019finite}. Motivated by the SED structure of Q-functions studied in the previous section, our adaptation of the LSTDQ method focuses on the estimation of truncated individual Q-functions, relying solely on locally available information. The actor then uses gradient descent to find a new policy. Since the actor and the critic are separate architectures we discuss them individually before providing the full algorithm in Section~\ref{sec:actor_critic}.

\subsection{Critic Architecture: Distributed Q-function Estimation}
\label{sec:critic}
Inspired by~Corollary~\ref{corr:H_trunc_err}, we adapt the LSTDQ algorithm to learn truncated individual Q-functions in a distributed way. 
%
%
For this purpose, let $\phi(x,u) := \text{svec} \left(   \left(\begin{smallmatrix}
x\\ u
\end{smallmatrix}\right) \left(\begin{smallmatrix}
x\\ u
\end{smallmatrix}\right)^{\top} \right)$. 
The Q-function can then be written as $Q^K(x,u) = \text{svec}(H)^\top \phi(x,u)$. Similarly, the truncated individual Q-function satisfies
\begin{align} \label{eq:trunc_ind_Q}
    \trunc{Q}_i^K(x,u) := \text{svec}(\trunc{H}_i^\kappa)^\top \phi(x,u) = h_i^{\top} \phi(x_{\mathcal{N}^\kappa_i}, u_{\mathcal{N}^\kappa_i}),
\end{align}

The goal of the critic is to estimate the parameters~$h_i$ for a policy~$K \in \mathcal{K}^\kappa$. LSTDQ is off-policy and we consider inputs of the form
\begin{equation*}
    u(t) := K_0x(t) + \eta(t), \ \eta(t) \sim \mathcal{N}(0,\sigma_\eta^2I),
\end{equation*}
where~$K_0 \in \mathcal{K}^\kappa$ is a stabilizing initial policy and~$\eta$ injected noise in order to ensure sufficient exploration. Now, assume agent~$i$ has collected a trajectory of~$T_c$ samples,~$D_{\mathcal{N}_i^\kappa} := \{( x_{\mathcal{N}^\kappa_i}(t), u_{\mathcal{N}^\kappa_i}(t)\}_{t=1}^{T_c+1}$ using~$u$ and, at each timestep, also communicated its current state and action with its neighbors.\footnote{See Algorithm 1  in~\cite[Section 5.2]{9137268} for an algorithm that collects these samples in a distributed manner.} Then, let~$u^K(t) := Kx(t)$ be the input from the controller for which we want to find the individual Q-function and, analogous to~\cite{krauth2019finite}, we define the LSTDQ estimator for the individual Q-functions by
\begin{subequations}
\label{eq:lstdqi}
\begin{equation} \label{eq:lstdqi_est}
        \est{h}_i \!:=\! \left(\sum_{t=1}^{T_c} \phi_i(t) \! \left( \phi_i(t) \!- \! \psi_i(t+1) \!+\! f_i \right) ^{\top} \! \right)^\dagger \! \sum_{t=1}^{T_c} \phi_i(t)c_i(t),
\end{equation}
wherein
\begin{align}
\phi_i(t) &:= \phi(x_{\mathcal{N}^\kappa_i}(t),u_{\mathcal{N}^\kappa_i}(t)), \label{eq:phi_i_def} \\ 
\psi_i(t) &:= \phi(x_{\mathcal{N}^\kappa_i}(t), u^{K}_{\mathcal{N}^\kappa_i}(t)),\label{eq:psi_i_def} \\ 
c_i(t) &:= x_i(t)^{\top}[S]_{ii}x_i(t) + u_i(t)^{\top}[R]_{ii}u_i(t),
\label{eq:c_i_def}\\
f_i &:= \text{svec} \left (\sigma ^2_w \begin{pmatrix}
I\\ 
[K]_{\mathcal{N}^\kappa_i :}
\end{pmatrix} \begin{pmatrix}
I\\ 
[K]_{\mathcal{N}^\kappa_i :}
\end{pmatrix} ^{\top} \right ). \label{eq:f_i_def} 
\end{align}
\end{subequations}
\noindent Here $(\cdot)^\dagger$ denotes the Moore-Penrose pseudo-inverse. We also exploit that~$H_i$ is symmetric and positive semidefinite by projecting  the estimate onto the set of symmetric positive semidefinite matrices using
\begin{equation} \label{eq:proj}
    \text{Proj} (\cdot) := \underset{X=X^{\top}, X \succeq 0}{\arg \min} ||X - \cdot||_F.
\end{equation}
We can now formulate the critic architecture in Algorithm~\ref{alg:critic_estQ} for agent~$i$.

\begin{algorithm}[htbp]
\caption{Critic: \textsc{EstimateQ}}\label{alg:critic_estQ}
 \begin{algorithmic}[1]
\Require Data trajectory~$D_{\mathcal{N}_i^\kappa}$; Current controller~$K_k \in \mathcal{K}^\kappa$; Trajectory length~$T_c$; Neighborhood size~$\kappa$.
\Procedure{EstimateQ}{$D_{\mathcal{N}_i^\kappa}$,~$K$,~$T_c$,~$\kappa$}

\State Use~$D_{\mathcal{N}_i^\kappa}$ to calculate~$\{u^{K_k}_{\mathcal{N}_i^\kappa}\}_{t=1}^{T_c+1}$ where~$u_i^{K_k}(t) = \sum_{j \in \mathcal{N}_i^\kappa} [K_k]_{ij}x_j(t)$.


\State Estimate~$\est{h}_i$ using~\eqref{eq:lstdqi}.

\State Project~$\text{smat}(\est{h}_i)$ using~\eqref{eq:proj} to form~$\est{H}_i$.

\State Share~$\est{H}_{i12}$ and~$\est{H}_{i22}$ with the~$\kappa$-neighborhood.
\State Form~$[\trunc{\est{H}}]_{ij} = \sum_{l \in \mathcal{N}_i^\kappa} [\trunc{\est{H}}_l]_{ij}$ for all~$j$ using~\eqref{eq:global_H11_ij}.

\EndProcedure
\end{algorithmic}
\end{algorithm}

After running Algorithm~\ref{alg:critic_estQ}, all agents have access to the relevant estimated parameters of the global truncated Q-function. This is due to the fact that the submatrices of~$\trunc{\est{H}}_i$ parameterizing~$\est{\trunc{Q}}_i^K(x_{\mathcal{N}^\kappa_i},u_{\mathcal{N}^\kappa_i})  = \begin{pmatrix}
x^{\top} & u^{\top}
\end{pmatrix} \trunc{\est{H}}_i
\begin{pmatrix}
x\\ u
\end{pmatrix}$, by construction, are such that $[\trunc{\est{H}}_{i11}]_{lj} = 0$ if $i \notin \mathcal{N}_l^\kappa \cap \mathcal{N}_j^\kappa$ and analogously for~$\trunc{\est{H}}_{i12}$ and~$\trunc{\est{H}}_{i22}$. This implies that
\begin{equation}  \label{eq:global_H11_ij}
    [\trunc{\est{H}}_{11}]_{lj} := \sum_{i=1}^N[\trunc{\est{H}}_{i11}]_{lj} = \sum_{i \in \mathcal{N}_l^\kappa \cap \mathcal{N}_j^\kappa}[\trunc{\est{H}}_{i11}]_{lj},
\end{equation}
which, again, also holds for~$\trunc{\est{H}}_{12}$ and~$\trunc{\est{H}}_{22}$. Meaning, the global truncated Q-function estimate can be recovered, even when each agent only communicates with its $\kappa$-neighborhood. Since~$\mathcal{N}_l^\kappa \cap \mathcal{N}_j^\kappa = \emptyset$ when~$j \notin \mathcal{N}_l^{2\kappa - 1}$, it is also clear from~\eqref{eq:global_H11_ij} that~$\trunc{\est{H}}$ is sparse in the sense that
\begin{equation} \label{eq:H11_12_22_zeros}
    [\trunc{\est{H}}_{11}]_{lj} = 0,~
    [\trunc{\est{H}}_{12}]_{lj} = 0,~
    [\trunc{\est{H}}_{22}]_{lj} = 0
    \text{ if } j \notin \mathcal{N}_l^{2\kappa - 1}.
\end{equation}
With the critic architecture in place, we now turn our attention to the actor architecture.

\subsection{Actor Architecture: Distributed Policy Update}
\label{sec:actor}

The role of the actor is to update the control strategy for the $\kappa$-truncated controller using the estimated Q-function from the critic. In our algorithm, the actor performs an (approximate) policy gradient descent procedure as described below.

We consider the gradient of the cost function $J(K)$ using the deterministic policy gradient theorem~\cite{silver2014deterministic}. The theorem states that the policy gradient is the expected gradient of the Q-function, i.e., $\nabla_K J(K) = \mathbb{E}_w \left[ \nabla_K Q^K(x,Kx)) \right]$, the gradient can then be estimated using online samples and calculating the gradient of the Q-function: 
\begin{equation} \label{eq:nabla_K_Q}
    \nabla_KQ^K(x,Kx) = 2(H^{\top}_{12} + H_{22}K)xx^{\top}.
\end{equation}
We note that the gradient depends on global information and set out to show that, in our setting, the partial derivative~$\partial Q^K(x,Kx) / \partial [K]_{ij}$ can be approximated by agent~$i$ only using local information.
First of all, the matrices~$H_{12}$ and~$H_{22}$ are not known and have to be replaced with the estimates~$\trunc{\est{H}}_{12}$ and~$\trunc{\est{H}}_{22}$. Equation~\eqref{eq:H11_12_22_zeros} tells us that these estimates are sparse. Combining this knowledge with the fact that~$K \in \mathcal{K}^\kappa$, we get
\begin{equation} \label{eq:H22_K_ij}
    [\trunc{\est{H}}_{22}K]_{ij} = \sum_{l=1}^N[\trunc{\est{H}}_{22}]_{il}[K]_{lj} = 0 \text{ if } j \notin \mathcal{N}_i^{3\kappa-2}.
\end{equation}
Using this in~\eqref{eq:nabla_K_Q} with the definition of partial derivative gives
\begin{equation} \label{eq:dQdKij}
    \frac{\partial \est{Q}^K(x,Kx)}{\partial [K]_{ij}} \! = \! 2\hspace{-8pt}\sum_{l \in \mathcal{N}_i^{2\kappa - 1}}\hspace{-8pt}[\trunc{\est{H}}_{12}^{\top}]_{il}[xx^{\top}]_{lj} + 2 \hspace{-8pt}\sum_{l \in \mathcal{N}_i^{3\kappa-2}}\hspace{-8pt}[\trunc{\est{H}}_{22}K]_{il}[xx^{\top}]_{lj}.
\end{equation}

From~\eqref{eq:global_H11_ij}, we see that in order to update agent~$i$'s parameters using~\eqref{eq:dQdKij} it is sufficient for agent~$i$ to have access to~$\trunc{\est{H}}_{j12}$ and~$\trunc{\est{H}}_{j22}$ for all~$j \in \mathcal{N}^\kappa_i$. Agent~$i$ also needs to calculate~$[\trunc{\est{H}}_{22}K]_{il}$ for~$l \in \mathcal{N}_i^{3\kappa-2}$ and from~\eqref{eq:H11_12_22_zeros} and~\eqref{eq:H22_K_ij}, we see that it needs access to~$[K]_{j:}$ for~$j \in \mathcal{N}_i^{2\kappa - 1}$. The final step in finding the gradient of the cost function is taking the expected value of the gradient in~\eqref{eq:nabla_K_Q}. The actor achieves this by sampling a trajectory and estimating the mean of the gradient over this trajectory. The actor's procedure for agent~$i$ is described in pseudocode in Algorithm~\ref{alg:actor}. As for the critic architecture, Algorithm~\ref{alg:actor} is run for each of the agents.

\begin{algorithm}[htbp]
\caption{Actor: \textsc{UpdateK}}\label{alg:actor}
 \begin{algorithmic}[1]
\Require Current controller~$K_k \! \in \mathcal{K}^\kappa$; Estimated Q-function parameters~$\trunc{\est{H}}$; Trajectory length~$T_a$; Step size~$\alpha$; Neighborhood size~$\kappa$.
\Procedure{UpdateK}{$K_k$,~$\est{H}_i$,~$T_a$,~$\alpha$,~$\kappa$}

\State Generate on-policy data~$\{x_i^a(t)\}_{t=1}^{T_a}$ by following policy~$K_k$. 
\State Share~$\{x_i^a(t)\}_{t=1}^{T_a}$ with the~$(3\kappa-2)$-neighborhood and share~$[K_k]_{i:}$ with the~$(2\kappa-1)$-neighborhood.
\State Estimate~$G_{ij} := \partial J(K) / \partial [K_k]_{ij}$ using the previously collected trajectory,~$\{x_i^a(t)\}_{t=1}^{T_a}$, and~\eqref{eq:dQdKij},

$\est{G}_{ij} := \frac{1}{T_a}\sum_{t=1}^{T_a} \frac{\partial \trunc{\est{Q}}^K(x^a(t),K_kx^a(t))}{\partial [K_k]_{ij}}.$
\State $[K_{k+1}]_{ij} \gets [K_k]_{ij} - \alpha \est{G}_{ij} \ \forall j \in \mathcal{N}_i^\kappa.$
\State \Return~$K_{k+1}$
\EndProcedure
\end{algorithmic}
\end{algorithm}

\subsection{Scalable Learning for Network LQR}
\label{sec:actor_critic}
We formulate Algorithm~\ref{alg:actor_critic} for distributed learning of network LQR by combining Algorithms~\ref{alg:critic_estQ} and \ref{alg:actor}.
\begin{algorithm}[htbp]
\caption{Scalable Learning for Network LQR}\label{alg:actor_critic}
 \begin{algorithmic}[1]
\Require Neighborhood size~$\kappa$; Initial stabilizing controller~$K_0 \in \mathcal{K}^\kappa$; Number of policy iterations~$k_{max}$; Critic trajectory length~$T_c$; Exploration variance~$\sigma_\eta^2$; Actor trajectory length~$T_a$; Actor step size~$\alpha$.
%
\State~$D_{\mathcal{N}_i^\kappa} \gets$ $\{( x_{\mathcal{N}^\kappa_i}(t), u_{\mathcal{N}^\kappa_i}(t)\}_{t=1}^{T_c+1}$ for all~$i$
\For{$k=0, \dots, k_{max}-1$}
    \State~$\trunc{\est{H}} \gets$ \textsc{EstimateQ}($D_{\mathcal{N}_i^\kappa}$,~$K_k$,~$T_c$,~$\kappa$) for all~$i$
    \State~$K_{k+1} \gets$ \textsc{UpdateK}($K_k$,~$\trunc{\est{H}}$,~$T_a$,~$\alpha$,~$\kappa$) for all~$i$
\EndFor
\State \Return~$K_{k_{max}}$
\end{algorithmic}
\end{algorithm}

Algorithm~\ref{alg:actor_critic} is an off-policy algorithm that returns a $\kappa$-truncated policy $K \in \mathcal{K}^\kappa$ after a set number of iterations specified by the user. It assumes access to an initial stabilizing controller. For each policy update, the actor has to collect a trajectory of length~$T_a$ to estimate the expected value, it is thus not offline even though the critic part of the algorithm is designed to run in an offline manner. In practice, $T_a$ can be small compared to $T_c$ making the online data collection negligible. The convergence properties of
Algorithm~\ref{alg:actor_critic} will be addressed in future work. Until then, the following section displays its performance in a simulation study.


\section{Numerical Simulation}
As an example problem, we consider the problem of controlling the temperature in a building with~$N=25$ rooms arranged in a 5-by-5 grid. We assume the thermal dynamics model studied in~\cite{zhang2022optimal} and~\cite{zhang2016decentralized}:
\begin{align*}
&\min_{\{u(t)\}}\int_{0}^{\infty}\sum_{i=1}^Ns_ix_i(t)^2 + u_i(t)^2, \\
&\dot{x}_i = \sum_{j \in \mathcal{N}_i, j \neq i} \frac{1}{v_i\zeta_{ij}}(x_j-x_i) +\frac{1}{v_i}u_i. 
\end{align*}
Here,~$v_i$ is the thermal capacitance of room~$i$,~$\zeta_{ij}$ the thermal resistance between two neighboring rooms and~$s_i$ the relative cost of deviating from the desired temperature. We assume~$\zeta_{ij} = 0.5 ^\circ$C / kW,~$v_i = 200 + 20 \times \mathcal{N}(0,1)$ kJ$/^\circ$C and~$s_i = 5$. 



We discretize the system with~$\Delta t = 1/4$ hour in the same way as in~\cite{zhang2022optimal}. We set~$K_0 = -3I$,~$T_a = 10000$,~$\sigma_\eta = 10$,~$\alpha = 0.005$ and run Algorithm~\ref{alg:actor_critic} using the real, but truncated Q-function by replacing the E\textsc{stimateQ} procedure with the real, $\kappa$-truncated individual Q-functions (recall~\eqref{eq:trunc_ind_Q}). Finally, we run Algorithm~\ref{alg:actor_critic} with the E\textsc{stimateQ} procedure and use~$T_c = 100,000$ samples. The relative cost, when using both the real and estimated Q-function, for different values on~$\kappa$ can be seen in Fig.~\ref{fig:thermal_rel_cost}. 

\begin{figure}[tbp]
    \centering
    \subfloat[]{
       \includegraphics[width=0.48\linewidth]{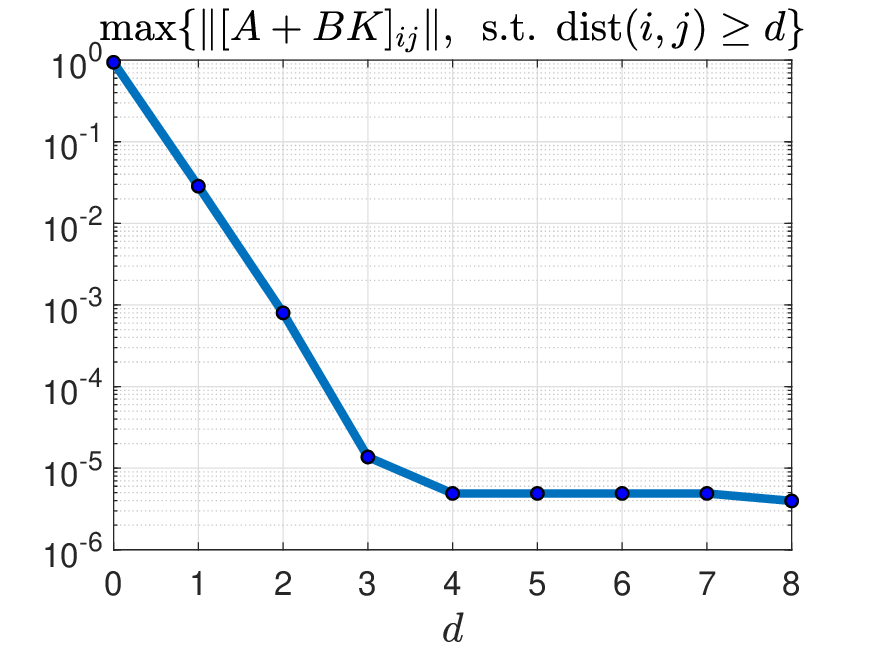}}
    \hfill
    \subfloat[]{
        \includegraphics[width=0.48\linewidth]{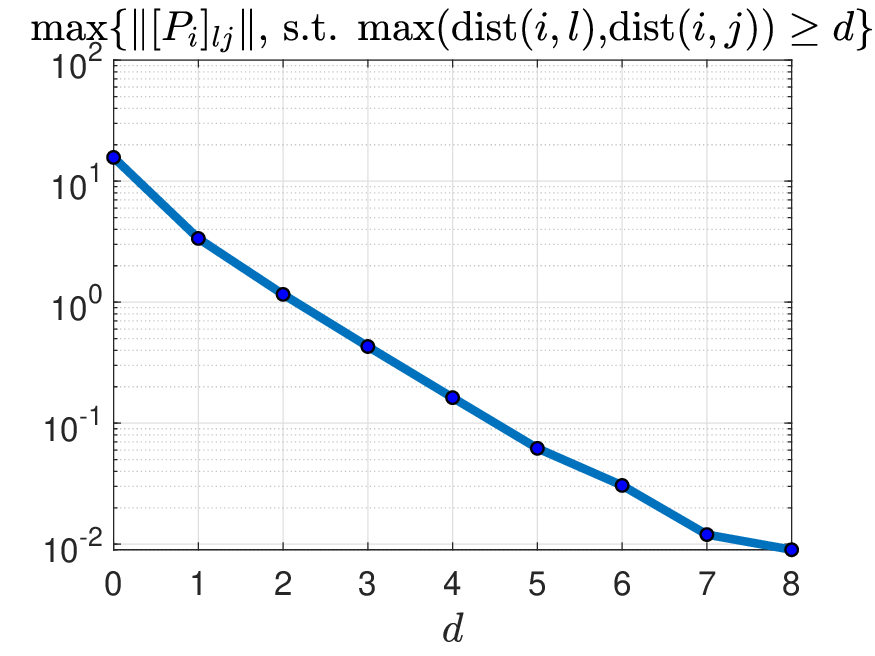}}
  \caption{ Spatial exponential decay of the closed-loop system in (a) and spatial exponential decay away from~$i$ for~$P_i$ that parameterizes the individual value function with~$i=1$ in (b). In both (a) and (b) $K=-3I$ was used.} 
  \label{fig:thermal_decay}
\end{figure}

\begin{figure}[tbp]
    \centering
    \subfloat[]{
        \includegraphics[width=0.48\linewidth]{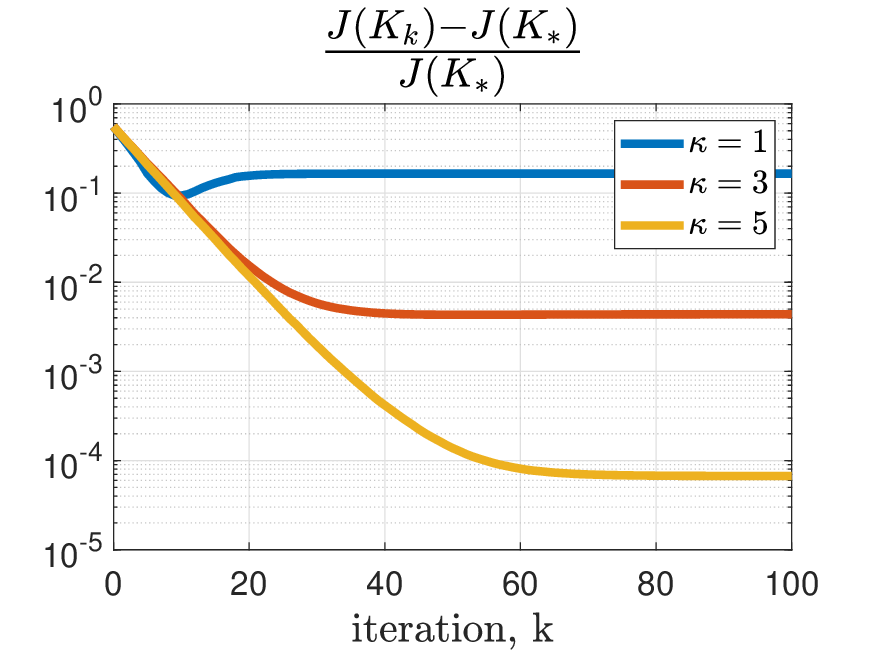}}
    \hfill
  \subfloat[]{
        \includegraphics[width=0.48\linewidth]{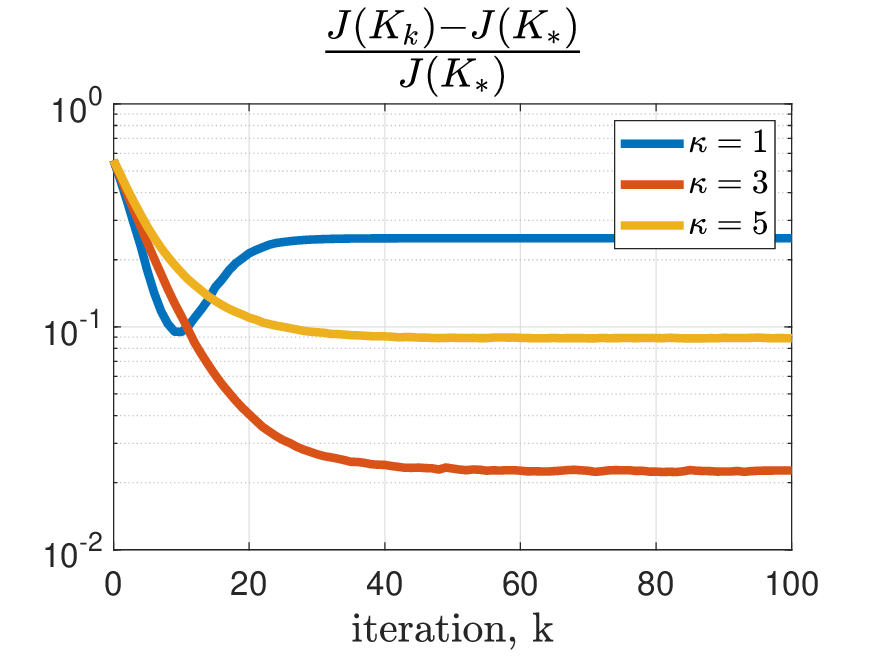}}
  \caption{Relative cost of learned thermal controllers. In (a), the real, but truncated, Q-function was used. In (b) the Qß-function was learned using~$T_c = 100,000$ samples.} 
  \label{fig:thermal_rel_cost}
\end{figure}

Fig.~\ref{fig:thermal_decay} displays the spatial decay in the closed-loop system and corroborates Theorem~\ref{thm:value_ASED} by illustrating how the individual value function also exhibits spatial exponential decay. Furthermore, Fig.~\ref{fig:thermal_rel_cost} shows Algorithm~\ref{alg:actor_critic}'s performance in practice both when sampling the truncated Q-function and when using the real truncated Q-function. We note that when sampling the Q-function with~$\kappa=3$, the algorithm performs better than with~$\kappa=5$. In theory, this should not happen as the extra parameters introduced when using~$\kappa=5$ always can be set to zero. This phenomenon stems from estimating the Q-function as it does not occur when using the real, truncated, Q-function (Fig.~\ref{fig:thermal_rel_cost}a). To mitigate this problem, more samples can be used in order to estimate the Q-function with higher accuracy.





\section{Conclusions and Future Work}

This paper concerns distributed learning of networked linear-quadratic controllers with decoupled costs and spatially exponentially decaying dynamics. We propose a scalable reinforcement learning algorithm that exploits the individual Q-functions' spatially exponentially decaying structure and numerically test it on a thermal control problem.

There are several directions for future work. Most important is, perhaps, proving convergence of Algorithm~\ref{alg:actor_critic} and stability of the resulting controllers. 
The question of sample complexity guarantees is also of interest. Finally, future work will include a practical evaluation of the algorithm through application to other problems in network learning and control.




\appendix 
\subsection{Auxiliary lemmas: Properties of SED and SED away from~$i$ matrices}
\label{apdx:aux}

\begin{lemma} \label{lemma:add_SED} Suppose ~$X,Y \in \mathbb{R}^{n \times m}$ are~$(c_x,\gamma _x)$-SED and~$(c_y,\gamma _y)$-SED, respectively, and let~$\gamma = \min(\gamma _x, \gamma _y)$. Then,~$X + Y$ is~$(c_x+c_y,\gamma)$-SED.
\end{lemma}

\begin{proof} 
\vspace{-1mm} 
    \begin{equation*}
        \left \| [X + Y]_{ij} \right \| \leq \left \| [X]_{ij} \right \| + \left \| [Y]_{ij} \right \| \leq (c_x+c_y) e^{-\gamma \dist(i,j)}.
        \vspace{-2mm} 
    \end{equation*}
\end{proof}

\begin{lemma}[Lemma 18 in~\cite{zhang2022optimal}] \label{lemma:18runyu}
Suppose~$X \in \mathbb{R}^{n \times m}$, $Y \in \mathbb{R}^{m \times p}$ are~$(c_x,\gamma _x)-SED$ and~$(c_y,\gamma _y)-SED$, respectively, and let~$\gamma = \min(\gamma _x, \gamma _y)$. Then,~$XY$ is~$(Nc_xc_y,\gamma)-SED$.
\end{lemma}
\begin{proof}
\vspace{-3mm}
    \begin{align*}
        & \left \| [XY]_{ij} \right \| = \left \| \sum_{r=1}^N [X]_{ir}[Y]_{rj}  \right \| \leq \sum_{r=1}^N ||[X]_{ir}||||[Y]_{rj}|| \\ &\leq \sum_{r=1}^N c_xc_ye^{-\gamma (\dist(i,r) + \dist(r,j))} \leq \sum_{r=1}^Nc_xc_ye^{-\gamma \dist(i,j)} \\ &= Nc_xc_ye^{-\gamma \dist(i,j)}.
    \end{align*}
\end{proof}
\begin{remark}
By following the proofs of Lemma~\ref{lemma:add_SED} and Lemma~\ref{lemma:18runyu}, it is easy to show that when $X$ and $Y$ are \emph{SED away from~$i$} instead of SED, then $X+Y$, $XY$ are also SED away from~$i$.
\end{remark}
Another property of matrices that fulfill Definition~\ref{def:SED_away_i} is that the decay away from~$i$ is preserved when such a matrix is multiplied by a SED matrix.

\begin{lemma}
    Let~$Y \in \mathbb{R}^{n \times m}$ be~$(c_y,\gamma_y)$-SED and let~$X \in \mathbb{R}^{m \times p}$ be~$(c_x,\gamma_x)$\textit{-SED away from i}. Furthermore, let $\gamma = \min(\gamma_x, \gamma_y)$. Then $XY$ is $(Nc_xc_y, \gamma)$\textit{-SED away from i}.
    \label{lemma:ASED}
\end{lemma}

\begin{proof}
    \begin{align*}
||[XY]_{lj}|| &= \left|\left| \sum_{r=1}^N [X]_{lr}[Y]_{rj} \right|\right| \leq \sum_{r=1}^N ||[X]_{lr}|| \ ||[Y]_{rj}|| \\ 
 &\leq \sum_{r=1}^N c_xc_y e^{-\gamma (\max(\dist(i,l),\dist(i,r)) + \dist(r,j))} \\
 &\leq Nc_xc_y e^{-\gamma \max(\dist(i,l),\dist(i,j))}.
\end{align*}
Where we used~$\max(a+c,b+c) \geq \max(a,b+c) \ \forall \ a,b,c \geq 0$ and the triangle inequality in the last step.
\end{proof}
\begin{remark}
    It is easy to see that the bound in Lemma~\ref{lemma:ASED} also holds for the product~$YX$, provided $X$ and $Y$ have suitable dimensions. 
\end{remark}


\subsection{Proof of Lemma~\ref{lemma:P_decay}}
\label{apdx:proof_P_decay_lemma}

Since~$L$ is stable, the solution to the Lyapunov equation is unique and given by $P = \sum_{k=0}^\infty (L^k)^{\top}ML^k$.
Define~$P^t$ to be the first~$t$ terms in the series $P^t := \sum_{k=0}^{t-1} (L^k)^{\top}ML^k$. Then, using that $L$ is~$(\tau,\rho)$-stable
\begin{align*}
        &\left \|  P - P^t \right \| \! = \! \left \| \sum_{k=t}^\infty (L^{\top})^kML^k \right \| \! \leq \! \sum_{k=t}^\infty ||(L^{\top})^k|| \ ||M|| \ ||L^k|| \\ 
        &\leq \!||M||\! \sum_{k=t}^\infty \! \tau^2 \! e^{-2\rho k} \! = \! ||M|| \! \sum_{k=0}^\infty \! \tau^2 \! e^{-2\rho (k+t)} \! = \! \frac{||M||\tau^2}{1-e^{-2\rho}}e^{-2\rho t}.
    \end{align*}
Now, since $L$ is~$(c_L,\gamma)$-SED, using some simple algebra on SED matrices (Lemma~\ref{lemma:18runyu} in Appendix~\ref{apdx:aux}),~$L^k$ is~$(N^{k-1}c_L^k,\gamma)$-SED and applying Lemma~\ref{lemma:ASED} twice, yields that~$(L^k)^{\top}ML^k$ is~$((Nc_L)^{2k}c_{M},\gamma)$-SED away from $i$. Furthermore,~$c_L \geq 1~$ implies~$ (Nc_L)^2 \geq 2$ and thus
\begin{equation*}
    \sum_{k=0}^{t-1}(Nc_L)^{2k}c_{M} = c_{M}\frac{(Nc_L)^{2t}-1}{(Nc_L)^{2}-1} \leq 2 (Nc_L)^{2(t-1)},
\end{equation*}
which means $P^t$ is~$(2c_{M}(Nc_L)^{2(t-1)}, \gamma)$-SED away from $i$.
Combining these results gives
\begin{align*}
    &\quad || [P]_{lj} || \leq || [P - P^t]_{lj} || + ||[P^t]_{lj}|| \\ &\leq ||P - P^t||+ 2c_{M}(Nc_L)^{2(t-1)}e^{-\gamma \max(\text{dist}(i,l),\text{dist}(i,j)) } \\ \! &\leq \!
     \frac{||M|| \tau^2}{1- e^{-2\rho}}e^{-2\rho t} \! + \! 2c_{M}(Nc_L)^{2(t-1)}e^{-\gamma \max(\text{dist}(i,l),\text{dist}(i,j)) }.
\end{align*}
This holds for any~$t$, in particular it holds when the two terms are roughly equal. That is, for~$t$ such that
\begin{equation*}
    e^{-2\rho t} = (Nc_L)^{2(t-1)}e^{-\gamma \max(\dist(i,l),\dist(i,j)) },
\end{equation*}
and we therefore set
\begin{equation*}
     t = \left \lfloor  \frac{\gamma \max(\dist(i,l), \dist(i,j)) }{2(\rho + \ln(Nc_L))} \right \rfloor + 1,
\end{equation*}
which gives
\begin{align*}
    || [P]_{lj} || &\leq
    \frac{||M|| \tau^2}{1- e^{-2\rho}}e^{-2\rho t} + \\ &2c_{M}(Nc_L)^{2(t-1)}e^{-\gamma \max(\dist(i,l),\dist(i,j)) } \\
&\leq  
\left( \frac{||M|| \tau^2}{1- e^{-2\rho}} + 2c_{M} \right ) \times \\ &\exp \left( -\frac{\rho \gamma }{\rho + \ln(Nc_L)} \max (\dist(i,l), \dist(i,j))  \right  ).
\end{align*}
$\qedsym$

\subsection{Proof of our main results: Theorem~\ref{thm:value_ASED} and Corollary~\ref{corollary:Q_ASED}}
\label{apdx:proof_V_Q_decay}

We first prove the following helper lemma:
\begin{lemma} \label{lemma:Si_KiR_Ki}
    Let~$K$ be~$(c_K,\gamma_{sys})$-SED. Then, the matrix~$S_i + [K]_{i:}^{\top} [R]_{ii}[K]_{i:}$ from the Lyapunov equation~\eqref{eq:P_i} is $(||[S]_{ii}|| +  ||[R]_{ii}||c_K^2, \gamma_{sys})$\textit{-SED away from i}.
\end{lemma}
\begin{proof}
First we note that
\begin{equation*}
    \left \|  [[K]_{i:}^{\top}[R]_{ii}[K]_{i:}]_{lj} \right \| \leq ||[R]_{ii}||c_K^2e^{-\gamma_{sys}\dist(i,l)} e^{-\gamma_{sys}\dist(i,j)},
\end{equation*}
and since~$e^{-\gamma_{sys}\dist(i,l)} \leq 1$ for all $l$ we get
\begin{align*}
     \left \| \![[K]_{i:}^{\top} \! [R]_{ii} \![K]_{i:}]_{lj} \!\right \| &\!\leq\! ||\![R]_{ii}\!|| c_K^2 \!\min(e^{\!-\!\gamma_{sys}\dist(\!i,l \!)}\!, e^{\!-\!\gamma_{sys}\dist(\!i,j\!)} \!) \\ &\!=\! ||[R]_{ii}||c_K^2 e^{-\gamma_{sys}\max(\dist(i,l),\dist(i,j))}.
\end{align*}
By construction of~$S_i$,~$[S_i]_{lj} = 0$ unless~$j=l=i$ and thus~$S_i$ is~$(||[S]_{ii}||, \gamma)$-SED away from~$i$, for any~$\gamma$. In particular it is true for~$\gamma = \gamma_{sys}$.
Combining these results using Lemma~\ref{lemma:add_SED} into
\begin{multline*}
   || [S_i + [K]_{i:}^{\top} [R]_{ii}[K]_{i:}]_{lj} || \leq \\ (||[S]_{ii}|| +  ||[R]_{ii}||c_K^2)   e^{-\gamma_{sys} \max(\dist(i,l), \dist(i,j))},
\end{multline*}
%
finishes the proof.
\end{proof}

Theorem~\ref{thm:value_ASED} and Corollary~\ref{corollary:Q_ASED} now follows from Lemmas~\ref{lemma:P_decay},~\ref{lemma:add_SED},~\ref{lemma:18runyu} and~\ref{lemma:Si_KiR_Ki}.

\begin{proof}[of Theorem \ref{thm:value_ASED}]
Lemma~\ref{lemma:add_SED} and~\ref{lemma:18runyu} gives that~$A+BK$ is~$(c_A + Nc_Bc_K, \gamma_{sys})$-SED and Lemma~\ref{lemma:Si_KiR_Ki} tells us the decay rate of~$S_i + [K]_{i:}^{\top} [R]_{ii}[K]_{i:}$. The result directly follows by setting~$L = A+BK$ and~$M = S_i + [K]_{i:}^{\top} [R]_{ii}[K]_{i:}$ in Lemma~\ref{lemma:P_decay}.
\end{proof}

\begin{proof}[of Corollary \ref{corollary:Q_ASED}]
    The result follows immediately by applying Lemma~\ref{lemma:add_SED} and~\ref{lemma:18runyu} on the definition of the submatrices and then choosing~$c_{H_i}$ as the maximum coefficient.
\end{proof}

\subsection{Proof of Corollary~\ref{corr:H_trunc_err}}
\label{apdx:proof_H_trunc_err}
    From the definition of~$\trunc{H}^\kappa_{i11}$ and Corollary~\ref{corollary:Q_ASED}, we have that~$||[H_{i11} - \trunc{H}^\kappa_{i11}]_{lj} || \leq c_{H_i}e^{-\gamma_{P_i}\kappa}.$ Multiplying by~$x \in \mathbb{R}^n$ with~$||x|| = 1$,
     \begin{align*}
         ||[(H_{i11} - \trunc{H}^\kappa_{i11})x]_l ||  &= \left| \left | \sum_{r=1}^N [(H_{i11} - \trunc{H}^\kappa_{i11})]_{lr} x_r \right| \right | \\ 
         &\leq c_{H_i}e^{-\gamma_{P_i}\kappa} \sum_{r=1}^N || x_r ||,
     \end{align*}
     and thus
     \begin{align*}
         ||(H_{i11} - \trunc{H}^\kappa_{i11})x ||^2 &\leq  N (c_{H_i}e^{-\gamma_{P_i}\kappa})^2  \left( \sum_{r=1}^N || x_r ||  \right )^2 \\ 
 &\leq N (c_{H_i}e^{-\gamma_{P_i}\kappa})^2 \sum_{r=1}^N || x_r ||^2 \\ 
 &=  N (c_{H_i}e^{-\gamma_{P_i}\kappa})^2 || x||^2.
     \end{align*}

Finally taking square roots on each side gives the result for~$H_{i11}$. Repeating the same procedure for~$H_{i12}$ and~$H_{i22}$ finishes the proof.
$\qedsym$

}

\bibliographystyle{IEEEtran}
\bibliography{IEEEabrv,acc_ref}

\addtolength{\textheight}{-12cm} 
\end{document}